\documentclass{article}

\usepackage[utf8]{luainputenc}
\usepackage{enumerate}
\usepackage{microtype}
\usepackage{framed}
\usepackage{hyperref}
\usepackage{bussproofs}
\usepackage{stmaryrd}
\usepackage{stackrel}
\usepackage{amsmath}
\usepackage{amssymb}
\usepackage{amsthm}

\usepackage{enumitem}

\newtheorem{theorem}{Theorem}
\newtheorem{proposition}[theorem]{Proposition}
\newtheorem{corollary}[theorem]{Corollary}
\theoremstyle{definition}
\newtheorem{definition}{Definition}

\usepackage{geometry}
\usepackage{pgf,tikz}
\usetikzlibrary{shapes,calc,
arrows,automata,positioning,backgrounds,fit,decorations.pathreplacing, matrix}
\usetikzlibrary {shapes.geometric,shapes.symbols}

\usepackage{todonotes}
\newcommand{\bra}{\langle}
\newcommand{\ket}{\rangle}

\newcommand{\mb}[1]{{\bf #1}}
\newcommand{\mc}[1]{{\mathcal{#1}}}

\newcommand{\Imp}{\mathop{\mathrm{Imp}}}

\newcommand{\Id}{\mathop{\mathrm{Id}}\nolimits}

\newcommand{\comp}{\mathrel{\mathrm{comp}}}

\newcommand{\vect}{\mathop{\mathrm{span}}}

\newcommand{\sas}{\mathbin{\&}}
\newcommand{\hook}{\mathbin{\delta}}

\newcommand{\fall}[1]{\forall \, {#1}, \ }
\newcommand{\fexist}[1]{\exists \, {#1}, \ }

\makeatletter
\newcommand{\verif@es}{\blacktriangleright}

\newcommand{\verifiesm}[1]{\mathrel{\verif@es_{{\scriptstyle \mathfrak {#1}}}}}
\newcommand{\verifiestarm}[1]{\mathrel{\verif@es^{\!\scriptstyle \#}_{{\!\scriptstyle \mathfrak {#1}}}}}
\makeatother

\makeatletter

\pgfdeclareshape{projector}{
	\savedanchor\northeast{%
        \pgfmathsetlength\pgf@x{\pgfshapeminwidth}%
        \pgfmathsetlength\pgf@y{\pgfshapeminheight}%
        \pgf@x=0.11\pgf@x
        \pgf@y=0.15\pgf@y
    }
	\savedanchor\southwest{%
        \pgfmathsetlength\pgf@x{\pgfshapeminwidth}%
        \pgfmathsetlength\pgf@y{\pgfshapeminheight}%
        \pgf@x=-0.11\pgf@x
        \pgf@y=-0.15\pgf@y
    }
	\inheritanchorborder[from=rectangle]

	\anchor{center}{\pgfpointorigin}
	\anchor{north}{\northeast \pgf@x=0pt}
    \anchor{east}{\northeast \pgf@y=0pt}
    \anchor{south}{\southwest \pgf@x=0pt}
    \anchor{west}{\southwest \pgf@y=0pt}
    \anchor{north east}{\northeast}
    \anchor{north west}{\northeast \pgf@x=-\pgf@x}
    \anchor{south west}{\southwest}
    \anchor{south east}{\southwest \pgf@x=-\pgf@x}
    \anchor{text}{
        \pgfpointorigin
        \advance\pgf@x by -.5\wd\pgfnodeparttextbox%
        \advance\pgf@y by -.5\ht\pgfnodeparttextbox%
        \advance\pgf@y by +.5\dp\pgfnodeparttextbox%
    }

	\backgroundpath{
		\southwest \pgf@xa=\pgf@x \pgf@ya=\pgf@y
        \northeast \pgf@yb=\pgf@y \pgf@xb=\pgf@x
		\pgfpathmoveto{\pgfpoint{\pgf@xa}{\pgf@ya}}
        \pgfpathlineto{\pgfpoint{\pgf@xa}{\pgf@yb}}
		\pgfpathmoveto{\pgfpoint{\pgf@xb}{\pgf@ya}}
        \pgfpathlineto{\pgfpoint{\pgf@xb}{\pgf@yb}}
		}
}

\tikzset{
	proj/.style={
		projector,
		draw, very thick,
		minimum width=2.2cm,
		minimum height=1.8cm},
smallproj/.style={
		projector,
		draw,
    minimum width=45pt,
    minimum height=30pt
	}
}

\makeatother

\begin{document}

\title{A Formal Theory for Finite-Dimensional Possibilistic Quantum Mechanics}
\author{Olivier Brunet \\ \texttt{olivier.brunet at normalesup.org}}

\maketitle

\begin{abstract}
	In this work, we present a logical formalism for reasoning about quantum systems in finite dimension. Contrary to the usual approach in quantum logic, our formalism is based classical first-order logic, which allows us to use the tools of model theory in our study. In particular, we show that our formal theory is complete, meaning that it entirely determines the behaviour of quantum systems. Moreover, we provide a characterization of the models of our formal theory, thus providing new insights in the study of hidden variable models of quantum theory.
\end{abstract}

\section{Introduction}

Quantum logic was born in a seminal article from Birkhoff and von Neumann~\cite{Birkhoff36QuantumLogic} who observed~that
\begin{quotation}
	propositional calculus of quantum mechanics has the same structure as an abstract projective geometry.
\end{quotation}
This departure from classical logic and from boolean algebras has initiated a whole field of research, with the emphasis put on the study of alternate algebraic structures such as orthomodular lattices \cite{Beltrametti81:LogicQM}, which has proven to be an important tool for gaining a better understanding of quantum mechanics and their foundations.

However, it is not necessary to depart from classical logic and in this article, we will present a formalism for reasoning about quantum systems based on classical logic. More precisely, we will define a first-order theory where formal variables correspond to quantum systems. Instead of being used for interpreting logical formulas, the closed subspaces of a Hilbert space (and, more generally, the elements of an orthomodular lattice) will provide predicates, called \emph{verification statements}, which express predictions regarding the detection of the particules of the quantum system.

This formulation is both non-probabilistic and operational: operational since logical terms are only used to describe the layout of a quantum circuit applied to a system, and non-probabilistic since our verifications statements will express falsifiable prediction about what can or cannot be detected at the output of a circuit.

Having a first-order logical theory, it is then possible to use the tools of model theory~\cite{Marker:Introduction} to study its properties. Moreover, the way this theory is formulated allows us to make a clear distinction between the epistemic description of a quantum system, the formulas of the language, and the ontic one, corresponding to the models of the theory and their elements.

The main contributions of this article are a formulation of finite-dimensional possibilistic quantum mechanics as a complete first-order theory, which axiomatization is, in its final form, given on figure~\ref{fig:rules2} on page~\pageref{fig:rules2}, and a characterization of its models.

In section~2, we will introduce the formalization of quantum circuits and their behaviours, define verification statements and provide a list of properties verified by these statements. Then, in section~3, we will turn these properties into a collection of axioms of a suitably defined first-order theory. Then, in the end of this section and in the next one, we will focus on some properties of this theory, and show that it is consistent and complete. Finally, we present an important consequence of our completeness proof, namely a characterization of the models of this theory.

Except for section~2 where we motivate the formilization of quantum circuits, this article will mainly focus on the logical aspects of the theory, leaving the discussion about quantum mechanics in a following companion article.

\section{Describing Quantum Systems and Their Evolution}


We will base our approach on a simplified formalization of quantum circuit with two types operators: unitary operators and projection operators. For instance, the circuit represented in figure~\ref{fig:bell} produces Bell state
\[|\Phi^+\ket = \frac 1 {\sqrt 2} \bigl(|00\ket + |11\ket\bigr)\]
where time flows from left to right and the projection on subspace~\(P\) is represented~as
\begin{center}
	\tikz{
		\matrix [column sep = {6mm,between origins}] {
			\node (a0) {};            &
			\node[proj] (a1) {\(P\)}; &
			\node (a2) {};              \\
		} ;
		\draw [-] (a0) -- (a1) -- (a2);
	}
\end{center}

\begin{figure}
	\begin{framed}
		\hbox{} \hfill \begin{tikzpicture}

			\tikzstyle{operator} = [draw,fill=white,minimum size=1.5em]
			\tikzstyle{phase} = [fill,shape=circle,minimum size=5pt,inner sep=0pt]
			\tikzstyle{point} = [fill,shape=circle,minimum size=3pt,inner sep=0pt]

			\matrix [column sep = {8mm,between origins}, row sep = {8mm,between origins}] {
				\node (a0) {} ;                                            &
				\node[proj] (a1) {\(\mb 0\)};                              &
				\node[operator] (a2) {H};                                  &
				\node[phase] (a3) {} ;                                     &
				\node (a4) {};                                               \\
				\node (b0) {};                                             &
				\node[proj] (b1) {\(\mb 0\)};                              &
				                                                           &
				\node[inner sep=0pt, outer sep = 0pt] (b2) {$\bigoplus$} ; &
				\node (b3) {} ;                                              \\
			} ;

			\draw [-] (a0) -- (a1) -- (a2) -- (a3) -- (a4) ;
			\draw [-] (b0) -- (b1) -- (b2) -- (b3) ;
			\draw [-] (a3) -- (b2);
		\end{tikzpicture} \hfill \hbox{}
		\caption{Creation of a Bell state}
		\label{fig:bell}
	\end{framed}
\end{figure}

A quantum system is said to be \emph{possible} if, when putting detectors are each output of the circuit, it is possible to simultaneously detect a particle at each detector. Otherwise, it is~\emph{impossible}. For instance, the output of the circuit depicted in figure~\ref{fig:bell2} is impossible since~\(| 10\ket\) is orthogonal to~\(|\Phi^+\ket\). Claiming that a circuit is~\emph{impossible} is a falsifiable statement, and will be at the heart of our approach.

\begin{figure}
	\begin{framed}
		\hbox{} \hfill \begin{tikzpicture}

			\tikzstyle{operator} = [draw,fill=white,minimum size=1.5em]
			\tikzstyle{phase} = [fill,shape=circle,minimum size=5pt,inner sep=0pt]
			\tikzstyle{point} = [fill,shape=circle,minimum size=3pt,inner sep=0pt]

			\matrix [column sep = {8mm,between origins}, row sep = {8mm,between origins}] {
				\node (a0) {} ;                                            &
				\node[proj] (a1) {\(\mb 0\)};                              &
				\node[operator] (a2) {H};                                  &
				\node[phase] (a3) {} ;                                     &
				\node[proj] (a4) {\(\mb 0\)} ;                             &
				\node (a5) {};                                               \\
				\node (b0) {};                                             &
				\node[proj] (b1) {\(\mb 0\)};                              &
				                                                           &
				\node[inner sep=0pt, outer sep = 0pt] (b2) {$\bigoplus$} ; &
				\node[proj] (b3) {\(\mb 1\)} ;                             &
				\node (b4) {} ;                                              \\
			} ;

			\draw [-] (a0) -- (a1) -- (a2) -- (a3) -- (a4) -- (a5) ;
			\draw [-] (b0) -- (b1) -- (b2) -- (b3) -- (b4) ;
			\draw [-] (a3) -- (b2);
		\end{tikzpicture} \hfill \hbox{}
		\caption{Creation of a Bell state}
		\label{fig:bell2}
	\end{framed}
\end{figure}

However, in the previous example,  we relied on some sort of~\emph{state} of a quantum state, a notion we explicitly want to avoid as a foundational element of our formalism. On the contrary, one objective of our approach is to investigate what constraints on any notion of quantum state are imposed by a logical approach.
However, through this example, we have introduced the different components of our formalism:
\begin{itemize}
	\item two types operators, namely unitary operators and projectors,
	\item a predicate telling whether a quantum system is impossible.
\end{itemize}
The semantics of these ingredients will entirely be specified by the five rules we will detail shortly. Before that, let us specify a few more elements.

First, we will consider a quantum system as a whole, so that the graph-like structure of circuit is replaced by a single linear composition of operators. For instance, the impossible circuit of figure~\ref{fig:bell2} becomes:
\begin{center}
	\begin{tikzpicture}[
			node distance=3mm,
			unitary/.style={
					rectangle,
					draw,
				},
			projo/.style={
					rectangle,
				},
			system/.style={
					rounded rectangle,
					draw,
				}
		]
		\node (0) [system] {\(\vphantom{Ip}s\)};
		\node (1) [projo, base right=of 0] {\(\vphantom{Ip} \mb 0 \otimes \mb 0\)};
		\node (2) [unitary, base right=of 1] {\(\vphantom{Ip} H \otimes \Id\)};
		\node (3) [unitary, base right=of 2] {\(\vphantom{Ip} \mathrm{CNot}(1, 2)\)};
		\node (4) [projo, base right=of 3] {\(\vphantom{Ip} \mb 1 \otimes \mb 0\)};
		\node (5) [base right=of 4] {\(\vphantom{Ip}: \Imp\)};
		\draw (0) -- (1) -- (2) -- (3) -- (4) -- (5) ;
		\draw (1.south west) -- (1.north west);
		\draw (1.south east) -- (1.north east);
		\draw (4.south west) -- (4.north west);
		\draw (4.south east) -- (4.north east);
	\end{tikzpicture}
\end{center}

where~\(s\) denotes the quantum system entering the circuit. This can be also written in a more mathematical form~as
\[ \mathop \pi\nolimits_{\mb 1 \otimes \mb 0} \circ \mathop {\mathrm{CNot}}\nolimits_{1, 2} \circ \mathop {H \otimes \Id} \circ \mathop \pi\nolimits_{\mb 0 \otimes \mb 0}(s) : \Imp\]

We also need introduce a few more notations, and one supplementary assumption stating that to each finite dimensional quantum system, one can associate a~\emph{dimension}~\(d \in \mb N\) so that projectors acting on this system can be associated to subspaces of~\(\mb C^d\). Let~\(\mb H_d\) denote the set of subspaces of~\(\mb C^d\), let~\(\mb V_{\! d}\) denote the subset of~\(\mb H_d\) made of the one-dimensional subspaces (or vector rays) of~\(\mb H_d\). To follow usual lattice notations,~\(\top\) (resp.~\(\bot\)) will denote the greatest (resp. smallest) element of~\(\mb H_d\) w.r.t.~inclusion. Finally let~\(\mb U_d\) denote the set of unitary operators acting on~\(\mb C^d\).

\medskip

We can now list the rules defining the meaning of both operators and the impossibility predicate in our approach.
First, the successive application of two projectors on orthogonal subspaces leads to an impossible circuit:

\begin{enumerate}[label=\textbf{Rule \arabic*}]
	\item \[ \fall {p, q \in \mb H_d}\qquad  q \leq p^\bot \quad \implies \quad
		      \begin{tikzpicture}[
				      baseline=(s.base),
				      node distance=3mm,
				      unitary/.style={
						      rectangle,
						      draw,
					      },
				      projo/.style={
						      rectangle,
					      },
				      system/.style={
						      rounded rectangle,
						      draw,
					      }
			      ]
			      \node (s) [system] {\(\vphantom{Ip}s\)};
			      \node (p1) [projo, base right=of s] {\(\vphantom{Ip} p\)};
			      \node (p2) [projo, base right=of p1] {\(\vphantom{Ip}q\)};
			      \node (imp) [base right=of p2] {\(\vphantom{Ip}: \Imp\)};
			      \draw (s) -- (p1) -- (p2) -- (imp);
			      \draw (p1.south west) -- (p1.north west);
			      \draw (p1.south east) -- (p1.north east);
			      \draw (p2.south west) -- (p2.north west);
			      \draw (p2.south east) -- (p2.north east);
		      \end{tikzpicture}
	      \]
	      Here, the circle labelled~\(s\) represented the input quantum system.
\end{enumerate}

If~\(q \leq p\), then projecting on~\(p\) and then on~\(q\) is equivalent to directly projecting on~\(q\). From this general principle, we derive two rules:
\begin{enumerate}[label=\textbf{Rule \arabic*}, resume]
	\item
	      \begin{multline*}
		      \fall {p, q \in \mb H_d} \quad q \leq p \implies \\
		      \Biggl(
		      \begin{tikzpicture}[
					      baseline=(s.base),
					      node distance=2mm,
					      unitary/.style={
							      rectangle,
							      draw,
						      },
					      projo/.style={
							      rectangle,
						      },
					      system/.style={
							      rounded rectangle,
							      draw,
						      }
				      ]
				      \node (s) [system] {\(\vphantom{Ip}s\)};
				      \node (p1) [projo, base right=of s] {\(\vphantom{Ip} p\)};
				      \node (p2) [projo, base right=of p1] {\(\vphantom{Ip}q\)};
				      \node (imp) [base right=5pt of p2] {\!\!\(\vphantom{Ip}: \Imp\)};
				      \draw (s) -- (p1) -- (p2) -- (imp);
				      \draw (p1.south west) -- (p1.north west);
				      \draw (p1.south east) -- (p1.north east);
				      \draw (p2.south west) -- (p2.north west);
				      \draw (p2.south east) -- (p2.north east);
			      \end{tikzpicture}
		      \iff
		      \begin{tikzpicture}[
					      baseline=(s.base),
					      node distance=2mm,
					      unitary/.style={
							      rectangle,
							      draw,
						      },
					      projo/.style={
							      rectangle,
						      },
					      system/.style={
							      rounded rectangle,
							      draw,
						      }
				      ]
				      \node (s) [system] {\(\vphantom{Ip}s\)};
				      \node (p2) [projo, base right=of s] {\(\vphantom{Ip}q\)};
				      \node (imp) [base right=5pt of p2] {\!\!\(\vphantom{Ip}: \Imp\)};
				      \draw (s) -- (p2) -- (imp);
				      \draw (p2.south west) -- (p2.north west);
				      \draw (p2.south east) -- (p2.north east);
			      \end{tikzpicture}
		      \Biggr)
	      \end{multline*}
	\item
	      \begin{multline*}
		      \fall {p, q, r \in \mb H_d} \quad q \leq p \implies \\
		      \Biggl(
		      \begin{tikzpicture}[
					      baseline=(s.base),
					      node distance=2mm,
					      unitary/.style={
							      rectangle,
							      draw,
						      },
					      projo/.style={
							      rectangle,
						      },
					      system/.style={
							      rounded rectangle,
							      draw,
						      }
				      ]
				      \node (s) [system] {\(\vphantom{Ip}s\)};
				      \node (p1) [projo, base right=of s] {\(\vphantom{Ip} p\)};
				      \node (q1) [projo, base right=of p1] {\(\vphantom{Ip}q\)};
				      \node (r1) [projo, base right=of q1] {\(\vphantom{Ip}r\)};
				      \node (imp) [base right=5pt of r1] {\!\!\(\vphantom{Ip}: \Imp\)};
				      \draw (s) -- (p1) -- (q1) -- (r1) -- (imp);
				      \draw (p1.south west) -- (p1.north west);
				      \draw (p1.south east) -- (p1.north east);
				      \draw (q1.south west) -- (q1.north west);
				      \draw (q1.south east) -- (q1.north east);
				      \draw (r1.south west) -- (r1.north west);
				      \draw (r1.south east) -- (r1.north east);
			      \end{tikzpicture}
		      \iff
		      \begin{tikzpicture}[
					      baseline=(s.base),
					      node distance=2mm,
					      unitary/.style={
							      rectangle,
							      draw,
						      },
					      projo/.style={
							      rectangle,
						      },
					      system/.style={
							      rounded rectangle,
							      draw,
						      }
				      ]
				      \node (s) [system] {\(\vphantom{Ip}s\)};
				      \node (p2) [projo, base right=of s] {\(\vphantom{Ip}q\)};
				      \node (r2) [projo, base right=of p2] {\(\vphantom{Ip}r\)};
				      \node (imp) [base right=5pt of r2] {\!\!\(\vphantom{Ip}: \Imp\)};
				      \draw (s) -- (p2) -- (r2) -- (imp);
				      \draw (p2.south west) -- (p2.north west);
				      \draw (p2.south east) -- (p2.north east);
				      \draw (r2.south west) -- (r2.north west);
				      \draw (r2.south east) -- (r2.north east);
			      \end{tikzpicture}
		      \Biggr)
	      \end{multline*}
\end{enumerate}

Let's now consider unitary operators. If~\(U(p) = q\), then ending by projecting on~\(p\) is equivalent to ending by applying~\(U\) and then projecting on~\(q\):
\begin{enumerate}[label=\textbf{Rule \arabic*}, resume]
	\item
	      \begin{multline*}
		      \fall {p, q \in \mb H_d} \fall {U \in \mb U_d} q = U(p) \implies \\
		      \Biggl(
		      \begin{tikzpicture}[
					      baseline=(s.base),
					      node distance=2mm,
					      unitary/.style={
							      rectangle,
							      draw,
						      },
					      projo/.style={
							      rectangle,
						      },
					      system/.style={
							      rounded rectangle,
							      draw,
						      }
				      ]
				      \node (s) [system] {\(\vphantom{Ip}s\)};
				      \node (u) [rectangle, draw, base right=of s] {\(\vphantom{Ip} U\)};
				      \node (p) [projo, base right=of u] {\(\vphantom{Ip}q\)};
				      \node (imp) [base right=5pt of p] {\!\!\(\vphantom{Ip}: \Imp\)};
				      \draw (s) -- (u) -- (p) -- (imp);
				      \draw (p.south west) -- (p.north west);
				      \draw (p.south east) -- (p.north east);
			      \end{tikzpicture}
		      \iff
		      \begin{tikzpicture}[
					      baseline=(s.base),
					      node distance=2mm,
					      unitary/.style={
							      rectangle,
							      draw,
						      },
					      projo/.style={
							      rectangle,
						      },
					      system/.style={
							      rounded rectangle,
							      draw,
						      }
				      ]
				      \node (s) [system] {\(\vphantom{Ip}s\)};
				      \node (p) [projo, base right=of s] {\(\vphantom{Ip}p\)};
				      \node (imp) [base right=5pt of p] {\!\!\(\vphantom{Ip}: \Imp\)};
				      \draw (s) -- (p2) -- (imp);
				      \draw (p.south west) -- (p.north west);
				      \draw (p.south east) -- (p.north east);
			      \end{tikzpicture}
		      \Biggr)
	      \end{multline*}
\end{enumerate}
Applying a unitary operator preserves the impossible character of a quantum system:
\begin{enumerate}[label=\textbf{Rule \arabic*}, resume]
	\item
	      \[
		      \fall {U \in \mb U_d} \qquad \Biggl(
		      \begin{tikzpicture}[
					      baseline=(s.base),
					      node distance=2mm,
					      unitary/.style={
							      rectangle,
							      draw,
						      },
					      projo/.style={
							      rectangle,
						      },
					      system/.style={
							      rounded rectangle,
							      draw,
						      }
				      ]
				      \node (s) [system] {\(\vphantom{Ip}s\)};
				      \node (imp) [base right=5pt of s] {\!\!\(\vphantom{Ip}: \Imp\)};
				      \draw (s) -- (imp);
			      \end{tikzpicture}
		      \iff
		      \begin{tikzpicture}[
					      baseline=(s.base),
					      node distance=2mm,
					      unitary/.style={
							      rectangle,
							      draw,
						      },
					      projo/.style={
							      rectangle,
						      },
					      system/.style={
							      rounded rectangle,
							      draw,
						      }
				      ]
				      \node (s) [system] {\(\vphantom{Ip}s\)};
				      \node (u) [rectangle, draw, base right=of s] {\(\vphantom{Ip} U\)};
				      \node (imp) [base right=5pt of p] {\!\!\(\vphantom{Ip}: \Imp\)};
				      \draw (s) -- (u) -- (imp);
			      \end{tikzpicture}
		      \Biggr)
	      \]
\end{enumerate}
The last rule, by far the most important, needs a bit more context. Consider two orthogonal vector rays~\(\psi_1, \psi_2 \in \mb V_{\! d}\) (with~\(\psi_2 \leq \psi_1^\bot\)) and assume that system~\(s\) is such that:
\[
	\begin{tikzpicture}[
			baseline=(s.base),
			node distance=2mm,
			system/.style={
					rounded rectangle,
					draw,
				}
		]
		\node (s1) [system] {\(\vphantom{Ip}s\)};
		\node (p1) [rectangle, base right=of s1] {\(\vphantom{Ip} \psi_1\)};
		\node (imp1) [base right=5pt of p1] {\!\!\(\vphantom{Ip}: \Imp\)};
		\draw (s1) -- (p1) -- (imp1);
		\draw (p1.north west) -- (p1.south west);
		\draw (p1.north east) -- (p1.south east);
	\end{tikzpicture}
	\qquad
	\mathbin{\mathrm{and}} \qquad \
	\begin{tikzpicture}[
			baseline=(s.base),
			node distance=2mm,
			system/.style={
					rounded rectangle,
					draw,
				}
		]
		\node (s) [system] {\(\vphantom{Ip}s\)};
		\node (p) [rectangle, base right=of s] {\(\vphantom{Ip} \psi_2\)};
		\node (imp) [base right=5pt of p] {\!\!\(\vphantom{Ip}: \Imp\)};
		\draw (s) -- (p) -- (imp);
		\draw (p.north west) -- (p.south west);
		\draw (p.north east) -- (p.south east);
	\end{tikzpicture}
\]
Suppose moreover that the projection on~\(\psi_1\) and~\(\psi_2\) is made using a path-based device such as a beam-splitter or a Stern-Gerlach apparatus. As a result, the systems projected on~\(\psi_1\) and on~\(\psi_2\) take two distinct branches, leading either to~\(A\) or to~\(B\) where one can put a detector.

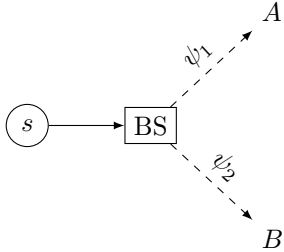
\begin{center}
	\begin{tikzpicture}
		\node (BS) [rectangle, draw] {BS};
		\node (S) [circle, draw, left=of BS] {\(s\)};
		\node (PSI1) [above right=of BS] {\(A\)};
		\node (PSI2) [below right=of BS] {\(B\)};
		\draw [-latex] (S) -- (BS);
		\draw [dashed, -latex] (BS) -- node [midway, sloped, above] {\(\psi_1\)} (PSI1);
		\draw [dashed, -latex] (BS) -- node [midway, sloped, above] {\(\psi_2\)} (PSI2);
	\end{tikzpicture}
\end{center}

Now, the assumption that~\(\pi_{\psi_1}(s) : \Imp\) (resp.~\(\pi_{\psi_2}(s) : \Imp\)) means that no particle can be detected, neither at~\(A\) nor at~\(B\). Moreover, as~\(A\) and~\(B\) can be placed arbitrarily far from one another, relativity dictates that the non-detection at~\(A\) cannot depend on whether the particle has been detected (or, even, if there is a detector) at~\(B\), and vice versa. In other words, no particle takes either path.

As a consequence, if one recombines the two paths (as in the following figure, where the two paths are deflected and recombined at~\(C\)), one obtains a projector on~\(\psi_1 \vee \psi_2\) and since no particle will be detected after~\(C\), we deduce that~\(\pi_{\psi_1 \vee \psi_2}(s) : \Imp\).

\begin{center}
	\begin{tikzpicture}
		\node (BS) [rectangle, draw] {BS};
		\node (S) [circle, draw, left=of BS] {\(s\)};
		\node (PSI1) [above right=of BS, rectangle, draw, inner sep=0pt, text width=20pt] {};
		\node (PSI2) [below right=of BS, rectangle, draw, inner sep=0pt, text width=20pt] {};
		\node (C) [below right=of PSI1] {\(C\)};
		\node (D) [right=40pt of C] {};
		\draw [-latex] (S) -- (BS);:
		\draw [dashed, -latex] (BS) -- node [midway, sloped, above] {\(\psi_1\)} (PSI1) ;
		\draw [dashed, -latex] (BS) -- node [midway, sloped, above] {\(\psi_2\)} (PSI2) ;
		\draw [dashed, -latex] (PSI1) -- (C) ;
		\draw [dashed, -latex] (PSI2) -- (C) ;
		\draw [dashed, -latex] (C) -- node [midway, above] {\(\psi_1 \vee \psi_2\)} (D) ;
	\end{tikzpicture}
\end{center}

\begin{enumerate}[label=\textbf{Rule \arabic*}, resume]
	\item Formally, this rule is expressed as
	      \begin{multline*}
		      \fall {\psi_1, \psi_2 \in \mb V_{\! d}} \qquad \psi_1 \leq \psi_2^\bot \implies \\[5pt]
		      \begin{tikzpicture}[
				      baseline=(s.base),
				      node distance=2mm,
				      system/.style={
						      rounded rectangle,
						      draw,
					      }
			      ]
			      \node (s1) [system] {\(\vphantom{Ip}s\)};
			      \node (p1) [rectangle, base right=of s1] {\(\vphantom{Ip} \psi_1\)};
			      \node (imp1) [base right=5pt of p1] {\!\!\(\vphantom{Ip}: \Imp\)};
			      \draw (s1) -- (p1) -- (imp1);
			      \draw (p1.north west) -- (p1.south west);
			      \draw (p1.north east) -- (p1.south east);
		      \end{tikzpicture}
		      \mathbin{\mathrm{and}} \
		      \begin{tikzpicture}[
				      baseline=(s.base),
				      node distance=2mm,
				      system/.style={
						      rounded rectangle,
						      draw,
					      }
			      ]
			      \node (s) [system] {\(\vphantom{Ip}s\)};
			      \node (p) [rectangle, base right=of s] {\(\vphantom{Ip} \psi_2\)};
			      \node (imp) [base right=5pt of p] {\!\!\(\vphantom{Ip}: \Imp\)};
			      \draw (s) -- (p) -- (imp);
			      \draw (p.north west) -- (p.south west);
			      \draw (p.north east) -- (p.south east);
		      \end{tikzpicture}
		      \implies
		      \begin{tikzpicture}[
				      baseline=(s.base),
				      node distance=2mm,
				      system/.style={
						      rounded rectangle,
						      draw,
					      }
			      ]
			      \node (s) [system] {\(\vphantom{Ip}s\)};
			      \node (p) [rectangle, base right=of s] {\(\vphantom{Ip} \psi_1 \vee \psi_2 \)};
			      \node (imp) [base right=5pt of p] {\!\!\(\vphantom{Ip}: \Imp\)};
			      \draw (s) -- (p) -- (imp);
			      \draw (p.north west) -- (p.south west);
			      \draw (p.north east) -- (p.south east);
		      \end{tikzpicture}
	      \end{multline*}
\end{enumerate}

These rules are summarized in figure~\ref{fig:crules}.

\begin{figure}
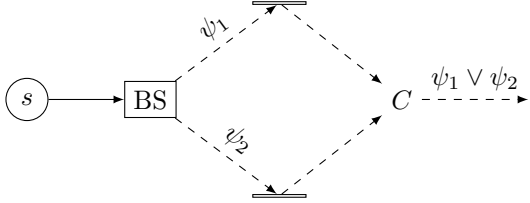

	\begin{gather*}
		q \leq p^\bot \implies \pi_q \circ \pi_p (s) : \Imp \label{c1} \tag{\(R_1\)}\\
		q \leq p \implies \bigl(\pi_q \circ \pi_p (s) : \Imp \iff \pi_q(s) : \Imp\bigr) \label{c2} \tag{\(R_2\)} \\
		q \leq p \implies \bigl(\pi_r \circ \pi_q \circ \pi_p (s) : \Imp \iff \pi_r \circ \pi_q(s) : \Imp\bigr) \label{c3} \tag{\(R_{3}\)} \\
		\pi_p(s) : \Imp \iff \pi_{U(p)} \circ U(s) : \Imp  \label{c4} \tag{\(R_4\)} \\
		s : \Imp \iff U(s) : \Imp \tag{\(R_5\)} \\
		\left.
		\begin{array}{c}
			\psi_1 \leq \psi_2^\bot \\[3pt]
			\pi_{\psi_1}(s) : \Imp  \\[3pt]
			\pi_{\psi_2}(s) : \Imp
		\end{array} \right\}
		\implies \pi_{\psi_1 \vee \psi_2}(s) : \Imp \label{c6} \tag{\(R_6\)}\\
	\end{gather*}

	\caption{Rules for circuits}
	\label{fig:crules}
\end{figure}


\bigskip

Using these rules, which deal with the transformation of circuits, we now define a collection of predicates on quantum system so as to reason about them.

\begin{definition}
	Given a quantum system~\(s\) of dimension~\(d\) and a property~\(p \in \mb H_d\), we said that~``\(s\) \emph{verifies}~\(p\)'' and write~\([s : p]\) if and only~if
	\[\fall {\varphi \in \mb V_{\! d}} \varphi \leq p^\bot \implies \pi_\varphi(s) : \Imp \]
\end{definition}

It must be remarked that verification statements are entirely based on \emph{impossibility}, so that they too are falsifiable and they do not refer to any sort of \emph{quantum state}.

\begin{figure}[ht]
	\begin{align*}
		          &  & \fexist x & \qquad  \neg [x : \bot] \label{ax:negbot} \tag{\(\neg \bot\)}                             \\
		          &  & \fall x   & \qquad  [x:\top] \label{ax:top} \tag{\(\top\)}                                            \\
		p \leq q  &  & \fall x   & \qquad [x:p] \rightarrow [x:q] \label{ax:leq} \tag{\(\leq\)}                              \\
		p \comp q &  & \fall x   & \qquad  [x:p] \wedge [x:q] \rightarrow [x : p \wedge q] \label{ax:wedge} \tag{\(\wedge\)} \\[6pt]
		          &  & \fall x   & \qquad  [x : p] \rightarrow [\pi_q(x) : p \sas q]       \label{ax:pii} \tag{\(\pi_i\)}    \\
		p \leq q  &  & \fall x   & \qquad [\pi_p \circ \pi_q (x) : \bot] \rightarrow
		[\pi_p(x) : \bot] \label{ax:pic} \tag{\(\pi_c\)}                                                                     \\
		          &  & \fall x   & \qquad  [\pi_q(x) : \bot] \rightarrow [x : q^\bot] \label{ax:pibot} \tag{\(\pi_\bot\)}    \\[3pt]
		          &  & \fall x   & \qquad  [x : p] \rightarrow [U(x) : U(p)] \label{ax:ui} \tag{\(u_i\)}                     \\
		          &  & \fall x   & \qquad  [U(x) : p] \rightarrow [x : U^{-1}(p)] \label{ax:ubot} \tag{\(u_e\)}
	\end{align*}
	\caption{The axioms of \(\mathrm{PQM}\)}
	\label{fig:rules}
\end{figure}

Using these predicates, we present in figure~\ref{fig:rules} a collection of various axioms regarding the verification of properties by quantum systems. Before reviewing them, let's observe that they are all consequences of our rules on quantum circuits:

\begin{proposition} \label{prop:valid_axioms}
	All the axioms listed in figure~\ref{fig:rules} following from the previous rules.
\end{proposition}

\begin{proof}
	The proof is developped in section~\ref{proof:valid_axioms}.
\end{proof}

The assertions of figure~\ref{fig:rules}, which will serve as axioms of our formal theory, can be described as follows:
\begin{itemize}
	\item The first one,~\ref{ax:negbot}, states that there are some systems that are not impossible, i.e. it is possible do detect a particle after projecting on a vector ray. This assertion is different in nature from the following ones.
	\item The next three axioms, namely~\ref{ax:top},~\ref{ax:leq} and~\ref{ax:wedge}, describe the set of properties verifies by a quantum system. For~\ref{ax:leq}, this applies to all~\(p, q \in \mb H_d\) such that~\(p \leq q\) and for~\ref{ax:wedge},~\(p\) and~\(q\) have to be compatible, that is the corresponding projectors commute (or, equivalently, one has~\(p = (q \wedge p) \vee (q^\bot \wedge p)\) using ortholattice notations).
	\item The next three relate the properties verified by a quantum system before and after being applied a projection.
	\item The last axiom relate the properties verifies by a quantum system before and after being applied a unitary operator.
\end{itemize}

Let us insist on the fact that in these axioms, the quantifications apply to quantum systems and not properties (that is elements of~\(\mb H_d\) and of~\(\mb V_{\! d}\)).

\section{A Formal Theory for Reasoning About Quantum Systems}

To further study the consequences of the axioms listed in figure~\ref{fig:rules}, let us now precisely define the whole logical framework we will use, in the form of a theory of first-order logic.


\begin{definition}
	For all~\(d \in \mb N^\star\), we define~\(\mc L_d\) as the formal language made of:
	\begin{description}
		\item [Functions] We have two types of functions:
		      \begin{enumerate}
			      \item For each~\(U \in \mb U_d\), a unary function~\(u\);
			      \item For some~\(p \in \mb H_d\), a unary function~\(\pi_p\) representing an orthogonal projector.
		      \end{enumerate}
		\item [Relations] For each~\(p \in \mb H_d\), a unary relation~\([ \, \cdot \, : p]\).
	\end{description}
\end{definition}

Regarding projectors, we do not impose a projector for each~\(p \in \mb H_d\): it is only required to have at least one projector on a vector ray, and one projector on a two-dimensional subspace. Without loss of generality, with an adequate use of unitary operators, we can consider that we have at least a projector for each vector ray, and for each two-dimensional subspace.


\bigskip

We recall that given a formal language~\(\mc L\), a~\(\mc L\)-structure~\(\mc M\) consists of a set~\(\mb M\) called the \emph{domain} of the structure, and an interpretation of each function and relation of the language: for each formal function~\(f\) of arity~\(a\), a function~\(f^{\mc M} : \mb M^a \rightarrow \mb M\) and for each formal relation~\(R\) of arity~\(a\), a relation~\(R^{\mc M} \subseteq \mb M^a\).

Two~\(\mc L_d\)-structures will be central in our study the \emph{Hilbert model}\footnote{The term \emph{model} will become clear very soon.}~\(\mc H_d\) with domain~\(\mb H_d\) and the \emph{Vector model}~\(\mc V_d\) with domain~\(\mb V_{\! d} \cup \{\bot\}\). They are defined the following way:

\begin{itemize}
	\item Let us start with the Hilbert model~\(\mc H_d\). Its domain is the set~\(\mb H_d\) of subspaces of~\(\mb C^n\) and:
	      \begin{align*}
		      \fall {p \in \mb H_d} \fall{U \in \mb U_d}      &  & U^{\mc H_d}(p)     & \mathrel{:=} \bigl\{ U(v) \bigm| v \in p \bigr\} \\
		      \fall {p, q \in \mb H_d}                        &  & \pi_q^{\mc H_d}(p) & \mathrel{:=} p \sas q                            \\
		      \mathrm{and \qquad {}} \fall {p, q \in \mb H_d} &  & [p : q]^{\mc H_d}  & \mathrel{:=} p \leq q
	      \end{align*}
	      Here, the \emph{Sasaki projection} \(p \sas q := q \cap (q^\bot \oplus p)\) is the lattice-theoretic equi\-valent of the orthogonal projection: if~\(\Pi_Q\) denotes the orthogonal projector on~\(Q\),~then
	      \[ P \sas Q = \bigl\{\Pi_Q(v) \bigm| v \in P \bigr\} \]

	\item The Vector model~\(\mc V_d\) has domain~\(\mb V_{\! d}^\star \mathrel{:=} \mb V_{\! d} \cup \{\bot\}\) i.e.~the set~\(\mb V_{\! d}\) of vector rays (i.e. one-dimensional subspaces) of~\(\mb C^n\) together with the null subspace and~where:
	      \begin{align*}
		      \fall {P \in \mb V_{\! d}^\star} \fall{U \in \mb U_d}  &  & U^{\mc V_d}(P)     & \mathrel{:=} \bigl\{ U(v) \bigm| v \in P \bigr\} \\
		      \fall {P \in \mb V_{\! d}^\star} \fall {Q \in \mb H_d} &  & \pi_Q^{\mc V_d}(P) & \mathrel{:=} P \sas Q                            \\
		      \fall {P \in \mb V_{\! d}^\star} \fall {Q \in \mb H_d} &  & [P : Q]^{\mc V_d}  & \mathrel{:=} P \subseteq Q
	      \end{align*}
\end{itemize}

These two models are closely related, as suggested by the similarity of their definitions. We formalize this with the following proposition:
\begin{proposition}
	\(\mc V_d\) is a substructure of~\(\mc H_d\), i.e.~loosely speaking, it is the restriction of~\(\mc H_d\) to~\(\mb V_{\! d}^\star\).
\end{proposition}
\begin{proof}
	It can be remarked, in particular, that the image of a vector ray by a unitary operator is a vector ray, and its Sasaki projection has dimension at most~one.
\end{proof}


\bigskip

Now that the language is specified, the propositions listed in figure~\ref{fig:rules} form the axioms and axiom schemata of a first-order theory on~\(\mc L_d\) which we will cal~\emph{Possibilistic Quantum Mechanics} (or~PQM for short, or~\(\mathrm{PQM}_d\) if we need an explicit reference to the dimension). It is important to note that this first-order theory doesn't have equality, and that quantification applies to quantum systems only, and not to elements of~\(\mb H_d\).

\begin{proposition}
	For all~\(n \in \mb N\), both \(\mc H_d\) and \(\mc V_d\) are models of~\(\mathrm{PQM}_d\).
\end{proposition}

\begin{proof}
	We recall that a~\(\mc L\)-structure is a \emph{model} of a formal theory on~\(\mc L\) if each sentence of the theory is true with respect to the interpretation of the structure.

	We will not show that each axioms of~\(\mathrm{PQM}_d\) is valid in~\(\mc H_d\), but here are a few justifications. Regarding~\ref{ax:leq}, it translates~as
	\[ \fall {x \in \mb H_d} x \leq \bot \]
	For~\ref{ax:pibot}, it is a classical result in orthomodular lattices:
	\[ \fall {x \in \mb H_d} x \sas q = \bot \iff x \leq q^\bot \]
\end{proof}

The existence of models has, using G\"odel's completeness theorem, a first important consequence:
\begin{corollary}
	For all~\(d \in \mb N\), \(\mathrm{PQM}_d\) is consistent.
\end{corollary}

That means that it is not possible, in~\(\mathrm{PQM}_d\), to prove both a proposition and its negation, this formal theory is not self-contradicting. The existence of models will also play an important role in the next section, where we show that this theory is complete.

\section{Completeness of the formal theory}

We recall that a theory~\(T\) is \emph{complete} if it is consistent and if for every sentence~\(\sigma\), either~\(T \models \sigma\) or~\(T \models \neg \sigma\) (and consistency implies that one cannot have both). Intuitively, this means that the theory is entirely determined, no sentence has its meaning still ``open''.

Another equivalent characterization of the completeness of a theory is that all its models are elementary equivalent, meaning that they all agree on which sentences are true: if~\(\mc M\) and~\(\mc N\) are models of a complete theory,~then
\[ \forall \, \sigma, \ \mc M \models \sigma \iff \mc N \models \sigma \]



\begin{theorem} \label{theo:complete}
	For all~\(d \in \mb N\), if~\(d \geq 3\) then~\(\mathrm{PQM}_d\) is complete.
\end{theorem}

The rest of this article will be largely devoted to the proof of this theorem, and to some of its consequences. This will be done by showing the elementary equivalence between models in three steps: first, we will show that for any sentence~\(\sigma\), we have~\(\mc H_d \models \sigma \implies \mc V_d \models \sigma\). Then, for any sentence~\(\sigma\) and any model~\(\mc M\), \(\mc M \models \sigma \implies \mc H_d \models \sigma\) and finally \(\mc V_d \models \sigma \implies \mc M \models \sigma\).


But first, let us remark that our theory, besides having no equality, also has no constants and all its relations and functions are of arity one. As such, it is a close relative of monadic first-order logics. An important consequence of this property that in our proof of completeness, we can restrict ourselves to sentences of the form~\(\exists \, x, \varphi(x)\) where~\(\varphi(x)\) is itself of the form
\[ [\overline f_1(x) : p_1] \wedge \cdots \wedge [\overline f_n(x) : p_n] \wedge \neg [\overline g_1(x) : q_1] \wedge \cdots \wedge \neg [\overline g_m(x) : q_m] \]
where the~\(p_i\)'s and the~\(q_i\)'s are elements of~\(\mb H_d\), and the~\(\overline f_i\)'s and the~\(\overline g_i\)'s are compositions of finitely many unitary functions.

In the following, we will fix such a~\(\varphi(x)\).
Moreover, for any model~\(\mc M\) with domain~\(\mb M\) and any~\(m \in \mb M\), \(\varphi(m)\) will denote
\[ [\overline f_1^{\mc M}(m) : p_1]^{\mc M} \wedge \cdots \wedge [\overline f_n^{\mc M}(m) : p_n]^{\mc M} \wedge \neg [\overline g_1^{\mc M}(m) : q_1]^{\mc M} \wedge \cdots \wedge \neg [\overline g_m^{\mc M}(m) : q_m]^{\mc M} \]
In particular, the semantics of first-order logic states~that
\[ \mc M \models \exists \, x, \varphi(x) \iff \exists \, m \in \mb M, \ \mc M \models \varphi(m) \]

\subsection{From subspaces to vectors}

In~\(\mc H_d\) (and in~\(\mc V_d\) as it is a substructure of~\(\mc H_d\)), we have
\[ [\pi_q(m) : p]^{\mc H} \iff m \sas q \leq p \iff m \leq p \hook q \iff [m : p \hook q]^{\mc H}\]
where~\(p \hook q\) is defined as~\(q^\bot \vee \bigl(p \wedge q)\), following from classical properties of orthomodular lattices. Similarly, we have
\[ [ U(m) : p ]^{\mc H} \iff [m : U^{-1}(p)]^{\mc H} \]
so that without loss of generality, in order to show
\[ \mc H_d \models \exists \, x, \varphi(x) \implies \mc V_d \models \exists \,  x, \varphi(x)\]
we can consider that \(\varphi(x)\) is of the form
\[ \varphi(x) =  [x : p_1] \wedge \ \cdots \ \wedge [x : p_n] \wedge \neg [x : q_1] \wedge \ \cdots \ \wedge \neg [x : q_m] \]
where only verifications statements appear.

Now, suppose that~\(\mc H_d\models \exists \, x, \varphi(x)\) and let~\(r \in \mb H_d\) be such that~\(\mc H_d \models \varphi(r)\), that~is
\[ \forall \, i \in \llbracket 1, n \rrbracket, r \leq p_i \]
--- or, equivalently,~\(r \leq \bigwedge_{i \in \llbracket 1, n\rrbracket} p_i \) --- and
\[ \forall \, j \in \llbracket 1, m\rrbracket, \ r \not \leq q_j \]
By putting~\(p_\infty = \bigwedge p_i\), this means that for all~\(j\), \(p_\infty \wedge q_j\) is a strict subspace of~\(p_\infty\). Now, it is a well know fact that a vector space cannot be obtained as a finite or countable union of strict subspaces (it is a consequence, for instance, of the Baire Category Theorem). In other words, we have\footnote{Here, we switch to set-like notations because we momentarily view subspaces as sets of vectors.}
\[ \bigcup_j \, (p_\infty \cap q_j) \neq p_\infty \]
so that there exists a non-zero vector~\(u \in p_\infty \setminus \bigcup_j \, (p_\infty \cap q_j)\), and we then~have
\[ \mc V_d \models \varphi\bigl(\vect(u)\bigr)\]
We have thus shown that for all~\(\varphi\),
\[ \mc H_d \models \exists \, x, \varphi(x) \implies \mc V_d \models \exists \, x, \varphi(x) \]

\subsection{From any model to subspaces}


Let~\(\mc M\) be any model of~\(\mathrm{PQM}_d\). We are going to show that there exists a function~\(\kappa : \mb M \mapsto \mb H_d\) such that for all~\(m \in \mb M\), all~\(p \in \mb H_d\) and all compositions~\(\overline f\) of unitary operators and projections,
\[
	\mc M \models [\overline f^{\mc M}(m) : p]^{\mc M} \iff \mc H_d \models [\overline f^{\mc H_d}\bigl(\kappa(m)\bigr): p]^{\mc H}
\]
so that
\[
	\forall \, m \in \mb M, \ \mc M \models \varphi(m) \iff \mc H_d \models \varphi\bigl(\kappa(m)\bigr)
\]
and, consequently,
\[
	\mc M \models \exists \, x, \varphi(x) \implies \mc H \models \exists \, x, \varphi(x)
\]

First, for any~\(m \in \mb M\), we define the~\emph{filter} associated to~\(m\) as the~set
\[ F^{\mc M}(m) = \bigl\{ p \in \mb H_d \bigm| [m : p]^{\mc M} \bigr\} \]

\begin{proposition}
	For all~\(m \in M\), we have:
	\begin{enumerate}
		\item \(\top \in F^{\mc M}(m)\);
		\item \(\forall \ p, q \in \mb H_d, \ \bigl(p \in F^{\mc M}(m) \mathrel{\mathrm{and}} p \leq q\bigr) \implies q \in F^{\mc M}(m)\);
		\item \(\fall {p, q \in F^{\mc M}(m)} p \sas q \in F^{\mc M}(m)\).
	\end{enumerate}
\end{proposition}

\begin{proof}
	The first two statements follow directly from axioms~\((\top)\) and~\((\leq)\).
	The third one follows from~\((\leq)\) and~\((\wedge)\), and the fact that for all~\(p, q \in \mb H_d\),  since \(p \vee q^\bot\) and~\(q\) are compatible and one has~\(p \sas q = q \wedge (p \vee q^\bot)\).
\end{proof}

\begin{proposition} \label{prop:no_two_atoms}
	For all~\(m \in M\), if~\(F^{\mc M}(m) \cap \mb V\) contains at least two distinct elements, then necessarily~\(\bot \in F^{\mc M}(m)\).
\end{proposition}
\begin{proof} See \cite{Brunet07PLA,Brunet07Topo}. A simplified proof is presented in subsection~\ref{proof:no_two_atoms}.
\end{proof}

\begin{proposition} \label{prop:incompatible}
	In dimension~\(d \geq 3\), if~\(p, q \in F^{\mc M}(m)\) are incompatible, then neither~\(p\) nor~\(q\) is minimal in~\(F^{\mc M}(m)\).
\end{proposition}
\begin{proof} See proposition~8 in~\cite{Brunet15QPL}. The proof is presented in subsection~\ref{proof:incompatible}.
\end{proof}

\begin{theorem}\label{theo:kappa}
	If~\(d \geq 3\), for every model~\(\mc M\) of \(\mathrm{PQM}_d\) with domain~\(\mb M\), there exists a function~\(\kappa : \mb M \rightarrow \mb H_d\) such that for all~\(m \in \mb M\),
	\[ \fall {p \in \mb H_d} p \in F^{\mc M}(m) \iff \kappa^{\mc M}(m) \leq p \]
	or, equivalently:
	\[ \fall {p \in \mb H_d} [m : p]^{\mc M} \iff [\kappa(m) : p]^{\mc H_d} \]
\end{theorem}
\begin{proof}
	First, since we work in finite dimension,~\(F^{\mc M}(m)\) has at least one minimal element w.r.t.~inclusion. Suppose that~\(F^{\mc M}(m)\) has two distinct minimal elements~\(p\) and~\(q\). From proposition~\ref{prop:incompatible}, they have to be compatible. But then, one has~\(p \wedge q \in F^{\mc M}(m)\) which contradicts the fact that~\(p\) and~\(q\) are two distinct minimal elements of~\(F^{\mc M}(m)\). As a consequence,~\(F^{\mc M}(m)\) has a unique minimal element, i.e.\ it possesses a least element, which is the expected~\(\kappa^{\mc M}(m)\).
\end{proof}

\begin{proposition}
	For all models~\(\mc M\), the function~\(\kappa\) defined in the previous theorem realizes a strong~\(\mc L_d\)-morphism from~\(\mc M\) to~\(\mc H_d\), that is for all~\(m \in \mb M\),
	\begin{align}
		\forall \, p \in \mb H_d, \quad & [m : p]^{\mc M} \iff [\kappa(m) : p]^{\mc H_d}                            \\
		\forall \, q \in \mb H_d, \quad & \kappa\bigl(\pi^{\mc M}_q(m)\bigr) = \pi^{\mc H_d}_q\bigl(\kappa(m)\bigr) \\
		\forall \, U \in \mb U_d, \quad & \kappa\bigl(u^{\mc M}(m)\bigr) = u^{\mc H_d}\bigl(\kappa(m)\bigr)
	\end{align}
\end{proposition}

\begin{proof} Point~(1) has been established in theorem~\ref{theo:kappa}, as for all~\(m \in \mb M\) and~\(p \in \mb H_d\),
	\[ [m : p]^{\mc M} \iff p \in F^{\mc M}(m) \iff \kappa(m) \leq p \iff [\kappa(m) : p]^{\mc H_d}\]
	Let us show point~(2). For~\(m \in \mb M\) and~\(q \in \mb H_d\), from~\([m : \kappa(m)]^{\mc M}\), we deduce~\([\pi^{\mc M}_q(m) : \kappa(m) \sas q]^{\mc M}\) so that
	\[ \kappa \circ \pi_q^{\mc M}(m) \leq \kappa(m) \sas q\]
	Now, suppose that~\(\kappa \circ \pi_q^{\mc M}(m) < \kappa(m) \sas q\). We can then define~\(u \in \mb V_{\! d}\) such that
	\[ u \leq \bigl(\kappa(m) \sas q \bigr) \wedge \bigl(\kappa \circ \pi_q(m)\bigr)^\bot \]
	Since~\(u \leq \bigl(\kappa \circ \pi_q(m)\bigr)^\bot\) and~\([\pi_q(m) : \kappa \circ \pi_q(m)]\), it follows that
	\[ \bigl[\pi_u \circ \pi_q(m) : \bot \bigr]\]
	But as~\(u \leq \kappa(m) \sas q \leq q\), this implies~that
	\[ \bigl[\pi_u(m) : \bot \bigr]\]
	and hence~\([m : u^\bot]\). As a consequence, \(\kappa(m) \sas q \leq u^\bot \sas q\). But the compati\-bility of~\(u\) and of~\(q\) (as~\(u \leq q\)) implies that~\(u^\bot \sas q = u^\bot \wedge q\) so that we finally have
	\[ u \leq \kappa(m) \sas q \leq u^\bot \wedge q \leq u^\bot\]
	which is absurd. Point~(3) can be proven in a similar way.
\end{proof}


Since~\(\kappa\) is a strong morphism, for all~\(m \in \mb M\), we have
\[ [ \overline f^{\mc M}(m) : p]^{\mc M} \iff [ \kappa \circ \overline f^{\mc M}(m) : p]^{\mc H} \iff [\overline f^{\mc H_d} \circ \kappa(m) : p]^{\mc H_d} \]

\begin{corollary}
	\[ \mc M \models \exists \, x, \varphi(x) \ \implies \ \mc H_d \models \exists \, x, \varphi(x) \]
\end{corollary}
\begin{proof}
	Clearly, if	\( \mc M \models \varphi(m)\), then \(\mc H_d \models \varphi \bigl(\kappa(m)\bigr)\).
\end{proof}

\subsection{From Vect to any model}

The previous result cannot be used to show that for any sentence~\(\varphi\),
\[ \mc H_d \models \exists \, x, \varphi(x) \iff \mc M \models \exists \, x, \varphi(x) \]
since~\(\kappa\) is not onto in general. However, we have a weaker resultat, namely that~\(\mb V_{\!d} \subseteq \kappa(\mb M)\), as we will show next, and this will prove to be sufficient to establish the completeness of~\(\mathrm{PQM}_d\).

\begin{proposition} \label{prop:surj}
	For all~\(v \in \mb V_{\! d}\), there exists an element~\(m \in \mb M\) such that~\(\kappa(m) = v\).
\end{proposition}
\begin{proof} Let~\(v \in \mb V_{\! d}\). Following axiom~\ref{ax:negbot}, there exists an element~\(m \in \mb M\) such that~\(\kappa(m) \neq \bot\). Let moreover~\(w \in \mb V_{\! d}\) be such that~\(w \leq \kappa(m)\). Then
	\[ \kappa \bigl(\pi_{w}(m)\bigr) = \kappa(m) \sas w = w\]
	Now, let~\(U \in \mb U_d\) such that~\(U(w) = v\). Putting~\(m' = U \circ \pi_{w}(m)\), we~have
	\[ \kappa(m') = v \]
\end{proof}

\begin{corollary}
	\[\mc V_d \models \exists\, x, \ \Phi(x) \quad \implies \quad \mc M \models \exists \, x, \ \Phi(x) \]
\end{corollary}
\begin{proof}
	If~\(\mc V_d \models \exists\, x, \varphi(x)\), then let~\(v \in \mb V\) be such that~\(\mc V_d \models \varphi(v)\). From proposition~\ref{prop:surj}, there exists a~\(m \in \mb M\) such that~\(\kappa(m) = v\). But then
	\[ \mc M \models \varphi(m) \iff \mc H_d \models \varphi\bigl(\kappa(m)\bigr) \iff \mc V_d \models \varphi\bigl(\kappa(m)\bigr)\]
	so that~\(\mc M \models \exists \, x, \ \varphi(x)\).
\end{proof}

\subsection{Wrapping up}

We have shown that for any~\(n \geq 3\) and for any model~\(\mc M\) of \(\mathrm{PQM}_d\),
\begin{multline*}
	\mc M \models \exists \, x, \varphi(x) \implies \mc H_d \models \exists \, x, \varphi(x) \\
	\implies \mc V_d \models \exists \, x, \varphi(x) \implies \mc M \models \exists \, x, \varphi(x)
\end{multline*}
and this chain of implications can be generalized to any sentence so that we have established:
\begin{theorem}
	For all~\(n \geq 3\), \(\mathrm{PQM}_d\) is complete.
\end{theorem}

Finally, let us explore two consequences of this theorem. First, we can modify the axioms of the theory. For instance, in axiom~\ref{ax:wedge}, we can drop the requirement that~\(p\) and~\(q\) are compatible (because of completeness, it is sufficient to show that this holds in~\(\mc H_d\)):
\begin{multline*}
	\forall \, x, p, q \in \mb H_d, \quad [x : p]^{\mc H} \wedge [x : q]^{\mc H} \implies \\ x \leq p \mathrel{\mathrm{and}} x \leq q
	\implies x \leq p \wedge q \implies [x : p \wedge q]^{\mc H}
\end{multline*}
This highly non-trivial simplification is a consequence of the topology of~\(\mb H_d\) for~\(d \geq 3\).

\medskip

Another simplification applies to rules regarding projection operators. We have~\(\pi_q^{\mc H}(x) = x \sas q\) and we have seen earlier that
\[ x \sas q \leq p \implies x \leq p \hook q \]
which translates~as
\[ [\pi_q(x) : p] \rightarrow [x : p \hook q]\]
This axiom implies, and thus can replace, both~\ref{ax:pic} and~\ref{ax:pibot}, as for the latter,
\[ \bot \hook q = q^\bot \vee (q \wedge \bot ) = q^\bot \]
and for the former, if~\(p \leq q\), then~\(p^\bot \hook q = p^\bot\) and~\(p^\bot \sas p = \bot\) so~that
\[ [\pi_p \circ \pi_q(x) : \bot]^{\mc H} \implies x \leq p^\bot \implies [\pi_p(x) : \bot]^{\mc H}
\]
We present in figure~\ref{fig:rules2} a revised axiomatization of~\(\mathrm{PQM}_d\) for~\(d \geq 3\).

\begin{figure}[ht]
	\begin{align*}
		         &  & \fall x   & \qquad  [x:\top] \tag{\(\top\)}                                          \\
		         &  & \fexist x & \qquad  \neg [x : \bot] \tag{\(\neg \bot\)}                              \\
		p \leq q &  & \fall x   & \qquad [x:p] \rightarrow [x:q] \tag{\(\leq\)}                            \\
		         &  & \fall x   & \qquad  [x:p] \wedge [x:q] \rightarrow [x : p \wedge q] \tag{\(\wedge\)} \\[6pt]
		         &  & \fall x   & \qquad  [x : p] \rightarrow [\pi_q(x) : p \sas q]       \tag{\(\pi_i\)}  \\
		         &  & \fall x   & \qquad  [\pi_q(x) : p] \rightarrow [x : p \hook q] \tag{\(\pi_e\)}       \\[3pt]
		         &  & \fall x   & \qquad  [x : p] \rightarrow [U(x) : U(p)] \tag{\(u_i\)}                  \\
		         &  & \fall x   & \qquad  [U(x) : p] \rightarrow [x : U^{-1}(p)] \tag{\(u_e\)}
	\end{align*}
	\caption{The axioms of \(\mathrm{PQM}_d\), revised for~\(d \geq 3\)}
	\label{fig:rules2}
\end{figure}

\medskip

A second important consequence of our study is that the existence of the~\(\kappa\)-function actually constitutes a criterion for characterizing models of~\(\mathrm{PQM}_d\).
\begin{theorem} \label{theo:charac} 
	If~\(d \geq 3\), an~\(\mc L_d\)-structure~\(\mc M\) with domain~\(\mc M\) is a model of \(\mathrm{PQM}_d\) if and only if there exists a strong~\(\mc L_d\)-morphism~\(\kappa : \mb M \rightarrow \mb H_d\) such that~\(\kappa(\mb M) \neq \{\bot\}\).
\end{theorem}

\begin{proof}
	We have already shown that for every model~\(\mc M\) of~\(\mathrm{PQM}_d\), there exists a strong morphism~\(\kappa : \mb M \rightarrow \mb H_d\) verifying~\(\kappa(\mb M) \neq \{\bot\}\).

	Conversely, suppose that~\(\mc M\) is a~\(\mc L_d\)-structure and that there exists a strong morphism~\(\kappa : \mb M \rightarrow \mb H_d\) such that~\(\kappa(\mb M) \neq \{\bot\}\). Let us show that~\(\mc M\) is indeed a model of~\(\mathrm{PQM}_d\). Since~\(\kappa\) is a strong morphism, one has:
	\[ \forall \, m \in \mb M, \ \forall \, p \in \mb H_d, \quad [m:p]^{\mc M} \iff [\kappa(m) : p]^{\mc H} \iff \kappa(m) \leq p \]
	As a consequence, all the axioms of~\(\mathrm{PQM}_d\) are trivially verified. For instance, axiom~\ref{ax:wedge} reduces~to
	\[ \forall \, m \in \mb M, \ \forall \, p, q \in \mb H_d, \quad \kappa(m) \leq p \mathbin{\ \mathrm{and} \ } \kappa(m) \leq q \implies \kappa(m) \leq p \wedge q\]
	which proof is direct.
\end{proof}

\section{Conclusion and perspectives}

This article shows that the study of quantum mechanics through the use of a classical logic is relevant, and complements the usual approach of quantum logic based on the study of structures other than boolean algebras. Moreover, through the results of model theory, it offers new insights into the foundations of quantum mechanics by means of the models of the theory, which are characterized by theorem~\ref{theo:charac}. However, we will postpone this discussion, and more generally the consequence of this formal theory on the foundations of quantum mechanics, in a following companion article.

One the logical side, this formalism can be extended on several directions. First, we have only dealt with finite dimension, which played a crucial role in theorem~\ref{theo:kappa}. It is crucial to study how to modify the theory in order to also take into account infinite-dimensional separable Hilbert spaces. Another important direction of development is to integrate inside the formalism the possibility for building composite systems, through some sort of tensor product, thus removing the constraint of having a linear circuit. Similarly, the formalism could be modify in order to take into account the different possible paths a particle can take, while in its current version, only one path can be followed. The multiplicative and additive connectives of linear logic \cite{Girard87tcs} seem to constitute a natural inspiration for formalizing this.


\appendix

\section{Proofs}

\subsection{Proof of proposition~\ref{prop:valid_axioms}} \label{proof:valid_axioms}
\begin{description}
	\item [\eqref{ax:top}] This follows from the fact that~\(\bigl\{v \in \mb V_{\! d} \bigm| v \leq \top^\bot\} = \emptyset\).

	\item [\eqref{ax:negbot}] This is a reformulation of the fact that there exists at least one~\(v \in \mb V_{\! d}\) such that applying~\(\pi_v\) does not always lead to an \emph{impossible} circuit.

	\item [\eqref{ax:leq}] Let~\(v \in \mb V_{\! d}\) such that~\(v \leq q^\bot\). Since~\(p \leq q\), we have~\(q^\bot \leq p^\bot\) and, by transitivity,~\(v \leq p^\bot\).

	      As a consequence,~\([s : p]\) implies that~\(\pi_v(s) : \Imp\).

	\item[\eqref{ax:wedge}] Suppose that~\(p, q \in \mb H_d\) are compatible, and let~\(v \in \mb V_{\! d}\) be such that~\(v \leq (p \wedge q)^\bot\). The compatibility of~\(p\) and~\(q\) implies that
	      \[ (p \wedge q)^\bot = p^\bot \vee \bigl(q^\bot \wedge p\bigr)\]
	      Moreover, we have
	      \[ v \leq \bigl(v \sas p^\bot\bigr) \vee \bigl(v \sas (q^\bot \wedge p)\bigr)\]
	      and
	      \[ v \sas (q^\bot \wedge p) \leq \bigl(v \sas p^\bot\bigr)^\bot \]
	      Now, \(v \sas p^\bot \leq p^\bot\) and~\([s : p]\) imply~\(\pi_{v \sas p^\bot}(s) : \Imp\). Similarly, we have \(\pi_{v \sas (q^\bot \wedge p)}(s): \Imp\). As a result, using~\eqref{c6}, we have
	      \[ \pi_{(v \sas p^\bot) \vee (v \sas (q^\bot \wedge p))}(s) : \Imp\]
	      and finally, using~\eqref{c3}, we conclude that \(\pi_v(s) : \Imp\).

	\item[\eqref{ax:pii}]
	      Let~\(v \in \mb V_{\! d}\) be such that~\(v \leq (p \sas q)^\bot\). Equivalently, we have~\(v \sas q \leq p^\bot\). In particular, since~\([s : p]\), this implies~\(\pi_{v \sas q}(s) : \Imp\) and, using~\eqref{c2},
	      \[ \pi_{v \sas q} \circ \pi_q(s): \Imp\]
	      Now, since~\(v \sas q^\bot \leq q^\bot\), using~\eqref{c1}, we also have
	      \[ \pi_{v \sas q^\bot} \circ \pi_q(s): \Imp\]
	      Finally, since~\(v \leq (v \sas q) \vee \bigl(v \sas q^\bot\bigr)\), using~\eqref{c6}, we deduce that
	      \[ \pi_v \circ \pi_q(s) : \Imp \]

	\item[\eqref{ax:pibot}]
	      Let~\(v \leq (q^\bot)^\bot\). Using~\eqref{c2}, we have
	      \[ \pi_v \circ \pi_q(s) : \Imp \iff \pi_v(s) : \Imp \]
	      But as~\([\pi_q(s) : \bot]\), obviously,~\(\pi_v \circ \pi_q(s) : \Imp\), so that~\(\pi_v(s) : \Imp\).

	\item[\eqref{ax:ui}] Let~\(v \leq U(p)^\bot\). We have, using~\eqref{c4}:
	      \[ \pi_v \circ U(s) : \Imp \iff \pi_{U^{-1}(v)}(s) : \Imp\]
	      But~\(U^{-1}(v) \leq p^\bot \iff v \leq U(p)^\bot\), so that~\([s : p]\) implies~\(\pi_{U^{-1}(u)}(s) : \Imp\) and, finally, \(\pi_v \circ U(s) : \Imp\).

	\item[\eqref{ax:ubot}]
	      We want to prove that~\([U(x):p]\) implies~\([x:U^{-1}(p)]\).
	      Let~\(v \in \mb V_{\! d}\) such that~\(v \leq \bigl(U^{-1}(p)\bigr)^\bot\).
	      We want to prove that~\(\pi_v(s) : \Imp\). But thanks to rule~\eqref{c4},
	      \[ \pi_v(s) : \Imp \iff \pi_{U(v)} \circ U(s) : \Imp\]
	      Now, we have \( v \leq \bigl(U^{-1}(p)\bigr)^\bot \iff U(v) \leq p^\bot\)
	      and as~\([U(x) : p]\), it follows that~\(\pi_{U(v)} \circ U(s) : \Imp\) and, equivalently, that~\(\pi_v(s) : \Imp\)
\end{description}

\subsection{Proof of proposition~\ref{prop:no_two_atoms}}
\label{proof:no_two_atoms}
Let~\(P(a)\) denote the fact that~\(F\) contains two vectors~\(u_+\) and~\(u_-\) which coordinates are
\[ u_+ \begin{pmatrix} a \\ 0 \\ 1 \\ 0 \\ \vdots \\ 0\end{pmatrix} \qquad u_- \begin{pmatrix} -a \\ 0 \\ 1 \\ 0 \\ \vdots \\ 0\end{pmatrix} \]
in a suitably chosen orthonormal basis.

\begin{proposition}
	For all~\(a \in \mathopen ]0, 1\mathclose[\), if \(P(a)\), then~\(P(b)\) for all~\(b \in [a, f(a)]\) where
	\[ f : x \mapsto \frac x {\sqrt{1 - x^2}}\]
\end{proposition}
Note that property~\(P(1)\) means that~\(F\) contains two orthogonal non-null vectors so that by compatible intersection,~\(F\) contains~\(\bot\).

\begin{proof}
	Let~\(x, y \in \mb R\) and define~\(w \begin{pmatrix} x \\ y \\ 1 \\ 0 \\ \vdots \\ 0\end{pmatrix}\) (the coordinates being expressed in the orthonormal basis corresponding to property~\(P(a)\)).

	Because of axiom~\eqref{ax:leq}, as~\(\vect(u_+)\) (resp.~\(\vect(u_-)\)) is in~\(F\), then so is~\(\vect(u_+, w)\) (resp.~\(\vect(u_-, w)\)), and if they are compatible, then~\(\vect(w) \in F\), following axiom~\eqref{ax:wedge}.

	But it is clear that we have~\(\vect(u_\pm, w) = \vect(v_\pm, w)\) with
	\[
		v_\pm = u_\pm - \frac {\bra w | u_\pm\ket}{\bra w | w\ket} w
	\]
	Since~\(\bra w|v_+\ket = \bra w|v_-\ket = 0\), then \(\vect(v_+, w)\) and~\(\vect(v_-, w)\) are compatible if and only if~\(\bra v_+ | v_- \ket = 0\), that~is if and only~if
	\[ x^2 + (1 - a^2) y^2 = a^2 \]
	In other words, we have shown that if~\(x^2 + (1 - a^2) y^2 = a^2\), then~\(w \in F\). One can recognize the equation of an ellipse with semi-minor axis~\(a\) and semi-major axis~\(f(a)\),
	with
	\[ f : x \in [0, 1\mathclose[ \mapsto \frac x {\sqrt{1 - x^2}}\]

	As a consequence, \( \forall \,  b \in [a, f(a)], \ P(b) \).
\end{proof}


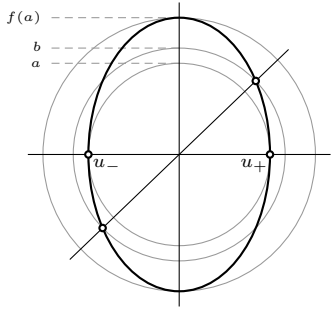
\begin{center}
	\begin{tikzpicture}
		\draw[black!40] (0, 0) ellipse (1.2cm and 1.2cm);
		\draw[black!40] (0, 0) ellipse (1.8cm and 1.8cm);
		\draw[black!40] (0, 0) ellipse (1.4cm and 1.4cm);
		\draw[thin, black!40, densely dashed] (0, 1.2) -- (-1.7, 1.2) node [left, black] {\(\scriptscriptstyle \vphantom{f(a)} a\)};
		\draw[thin, black!40, densely dashed] (0, 1.4) -- (-1.7, 1.4) node [left, black] {\(\scriptscriptstyle \vphantom{f(a)} b\)};
		\draw[thin, black!40, densely dashed] (0, 1.8) -- (-1.7, 1.8) node [left, black] {\(\scriptscriptstyle \vphantom{f(a)} f(a)\)};
		\draw [-] (-2, 0) -- (2, 0);
		\draw [-] (0, -2) -- (0, 2);
		\draw[thick] (0, 0) ellipse (1.2 and 1.8);
		\draw (1.44561264, 1.38210133) -- (-1.44561264, -1.38210133);
		\draw[fill=white, thick] (1.01192885, 0.96747093) circle (1.2pt);
		\draw[fill=white, thick] (-1.01192885, -0.96747093) circle (1.2pt);
		\draw[fill=white, thick] (1.2, 0) circle (1.2pt) node [below left=-2pt and -3pt] {\(\scriptstyle u_+\)};
		\draw[fill=white, thick] (-1.2, 0) circle (1.2pt) node [below right=-2pt and -2pt] {\(\scriptstyle u_-\)};
	\end{tikzpicture}
\end{center}

\begin{proof}[Proof of proposition~\ref{prop:no_two_atoms}]

	The proposition can be restated as the fact that for all~\(a \in \mathopen ]0, 1]\), \(P(a)\) implies~\(P(1)\). But the study of~\(f\) shows that for all~\(a \in \mathopen]0, 1]\), there exists an integer~\(n\) such that~\(f^{(n)}(a) \geq 1\), so that~\(F\) contains two orthogonal vectors, and hence contains~\(\bot\).
\end{proof}

\subsection{Proof of proposition~\ref{prop:incompatible}}
\label{proof:incompatible}
Let~\(P, Q\) be two incompatible elements of~\(F^{\mc M}(m)\), seen as subspaces, and denote~\(\Pi_P\) (resp.~\(\Pi_Q\)) the orthogonal projection~\(P\) (resp.~on~\(Q\)).

First, one can remark that the restriction of~\(\Pi_P \circ \Pi_Q\) to~\(P\) is self-adjoint with its eigenvalues in~\([0, 1]\). Moreover, since~\(P\) and~\(Q\) are not compatible, it possesses an eigenvalue~\(\lambda \in \mathopen] 0, 1\mathclose[\). Let~\(u\) be an associated eigenvector, and let~\(v = \Pi_Q(u)\). Obviously, \(\Pi_P(v) = \lambda u\), and also~\(u \not \in Q\) and~\(v \not \in P\).

Let's define~\(C = \vect(u, v)\). Having, equivalently,~\(C = \vect(u, v - \lambda u)\), since~\(u \in P\) and~\(v - \lambda u \in P^\bot\), it follows that~\(C\) and~\(P\) are compatible. Similarly,~\(C\) and~\(Q\) are compatible.

As a consequence, \(\bigl\{P \sas C, Q \sas C\bigr\} \in F^{\mc M}\bigl(\pi_C^{\mc M}(m)\bigr)\). But~\(P \sas C = P \wedge C = \vect(u)\) and~\(Q \sas C = Q \wedge C = \vect(v)\), so that~\(F^{\mc M}\bigl(\pi_C^{\mc M}(m)\bigr)\) contains two distinct elements of~\(V\). As a consequence, from proposition~\ref{prop:no_two_atoms}, it follows that~\(\bot \in F^{\mc M}\bigl(\pi_C^{\mc M}(m)\bigr)\). As a consequence from~\eqref{ax:pibot}, \([m : C^\bot]^{\mc M}\). Finally, using~\eqref{ax:wedge}, we have~\([m : P \wedge C^\bot]^{\mc M}\) with~\(P \wedge C^\bot < P\). This shows that~\(P\) is not minimal in~\(F^{\mc M}(m)\) and neither is~\(Q\).

\bibliographystyle{plain}


\end{document}